\documentclass[10pt,web]{ieeecolor}
\usepackage{generic}
\usepackage{cite}
\usepackage{mathtools,amsmath,amssymb,amsfonts,latexsym,mathrsfs}
\usepackage{nccmath}
\usepackage{bm}
\usepackage{subfig}
\usepackage{graphicx}
\usepackage{booktabs}
\usepackage{comment}

\graphicspath{{figures/},{figures/dashpot/},{figures/quantum/},{figures/dashpot2/}}

\newtheorem{theorem}{Theorem}
\newtheorem{corollary}{Corollary}

\newtheorem{remark}{Remark}

\newcommand{\ket}[1]{\left|#1\right\rangle}
\newcommand{\bra}[1]{\left\langle#1\right|}

\renewcommand{\vec}{\bm}

\newcommand{\Tr}{\operatorname{Tr}}

\newcommand{\A}{\bm A}

\renewcommand{\S}{\bm S}
\newcommand{\D}{\bm D}
\newcommand{\bPhi}{\bm \Phi}
\newcommand{\I}{\bm I}
\newcommand{\C}{\bm C}
\newcommand{\T}{\bm T}
\newcommand{\w}{\hat{\bm w}}
\renewcommand{\v}{\hat{\bm v}}
\newcommand{\e}{\hat{\bm e}}
\newcommand{\z}{\hat{\vec{z}}}
\newcommand{\Eta}{\hat{\bm \eta}}
\newcommand{\G}{\bm G}
\newcommand{\bDelta}{\bm \Delta}
\newcommand{\X}{\bm X}
\newcommand{\Z}{\bm Z}
\newcommand{\M}{\bm M}
\newcommand{\muSSV}{$\mu$SSV~}

\begin{document}
\title{Disturbance-agnostic robust performance with structured uncertainties and initial state error in classical versus quantum oscillatory systems}
\author{E.~Jonckheere, 
\IEEEmembership{Life Fellow, IEEE}
S.~G.~Schirmer, 
\IEEEmembership{Member, IEEE},
F.~C.~Langbein, 
\IEEEmembership{Member, IEEE},
C.~A.~Weidner, \IEEEmembership{Member, IEEE}, and 
S. O'Neil, \IEEEmembership{Member, IEEE}
\thanks{This research was supported in part by NSF Grant IRES-1829078.}
\thanks{EAJ and SON are with the Department of Electrical and Computer Engineering, University of Southern California, Los Angeles, CA 90089 USA (e-mail: jonckhee@usc.edu).}
\thanks{SGS is with the Faculty of Science \& Engineering, Swansea University, Swansea SA2 8PP, UK (e-mail: s.m.shermer@gmail.com).}
\thanks{FCL is with the School of Computer Science and Informatics, Cardiff University, Cardiff CF24 3AA, UK (e-mail: frank@langbein.org).}
\thanks{CAW is with the Quantum Engineering Technology Laboratories, University of Bristol, Bristol BS8 1FD, United Kingdom (e-mail: c.weidner@bristol.ac.uk).}
}
\maketitle

\begin{abstract}
A method to quantify robust performance for situations where structured parameter variations and initial state errors rather than extraneous disturbances are the main performance limiting factors is presented. The approach is based on the error dynamics, the difference between nominal and perturbed dynamics, driven by either the unperturbed or perturbed state, rather than an artificially imposed disturbance. The unperturbed versus perturbed dichotomy can be interpreted as the relative error dynamics scaled by either the unperturbed or perturbed dynamics. The error dynamics driven by unperturbed state has the unique feature of decoupling the effect of physically meaningful uncertainties from an additive disturbance. The perturbed case offers the possibility to side-step Structured Singular Value (SSV) computations. Applications to a lightly damped mechanical system and a slowly dephasing quantum system demonstrate the usefulness of the concepts across a broad range of systems.  Finally, a fixed-point algorithm specifically developed for quantum systems with state transitions depending in a nonlinear fashion on uncertainties is proposed as the substitute for classical SSV. 
\end{abstract}

\begin{IEEEkeywords}
Uncertain systems, robust control, H-infinity control, quantum control.
\end{IEEEkeywords}

\section{Introduction}\label{sec:intro}

Robust control is often formulated in terms of controlling parametrically uncertain Linear Time-Invariant (LTI) dynamics driven by extraneous disturbances. Robust performance can then be assessed by bounding the infinity-norm of the transmission of the disturbance to the output-to-be-controlled~\cite{Zhou,PackardDoyle,Safonov95,IQC}.
This raises the question whether the performance is limited by a physically present disturbance, or whether the disturbance is simply added as a convenient way to make the $\mu$ or Structured Singular Value (\muSSV) machinery applicable.  For example, the disturbances acting on a space structure are very light, even deterministic (gravity gradient), casting doubt on a robust performance based on the transmission of a ``disturbance" to some output~\cite{TRW}.  In addition, other potential sources of error such as in the initial state, not commonly considered in classical robustness, are relevant for certain applications.  This particularly applies to quantum control, where, e.g., the fidelity of quantum state transfer is degraded by uncertain initial state preparation errors~\cite{Carrie_Kosut}.  We propose an alternative robust performance formulation that is independent of extraneous disturbances, incorporates initial state errors quite naturally, and allows for LTI dynamics depending in a non-affine manner on the uncertain parameters.  

What distinguishes the present approach from classical robust performance is that the extraneous disturbance is suppressed by introducing a new state, the difference between the perturbed dynamics and the unperturbed dynamics.  The dynamics of this error has two different but equivalent state-space formulations, driven by either the unperturbed-dynamics state or the perturbed-dynamics state.  The annihilation of the extraneous disturbance allows the assessment of the performance solely in terms of the parametric uncertainties, independently of extraneous disturbances. Moreover, the frequency response of the error appears quite naturally scaled by the unperturbed or perturbed dynamics, leading to relative errors.  Computationally, the unperturbed-state scaling has the advantage of simplifying the frequency sweep typically required in \muSSV-analysis.  In the perturbed case, \muSSV-computations can be completely eliminated. The approach overcomes a general technicality, typical in the IQC signal analysis approach~\cite{IQC}, that the driving signals need to be in $L^2$. The ``disturbance-independent" approach was originally motivated by quantum control problems~\cite{robust_performance_open}.  Here, we broaden its remit to classical systems, indicate a way to sidestep the \muSSV-computations, and extend classical Linear Parametrically Varying (LPV) perturbations to nonlinear parametrically varying perturbations.

The paper is organized as follows. In Sec.~\ref{sec:error-dynamics}, the two formulations of the error dynamics relative to the perturbed and unperturbed system dynamics are presented. Robust performance with respect to the dynamics of the unperturbed system is developed in Sec.~\ref{sec:unperturbed}. In Sec.~\ref{sec:perturbed}, the performance is analyzed relative to the perturbed dynamics.  In Sec.~\ref{sec:mechanical_example} we apply the results to a simple mechanical system, and in Sec.~\ref{sec:quantum_examples} we consider a quantum mechanical example.

The present paper is an outgrowth of~\cite{robust_performance_open}, which was dedicated to open quantum systems. Here, we broaden the range of its applications to include classical systems, stressing a classical-quantum discrepancy.  

\section{Relative error dynamics}\label{sec:error-dynamics}

The system dynamics resulting from the model are referred to as the \emph{unperturbed} system dynamics. The dynamics of the real system, which deviates from the model dynamics, is referred to as the \emph{perturbed} system dynamics. Formulating the dynamics in terms of the controllable state-space representations, where $\vec{r}_u(t)$ is the state of the unperturbed system (following nominal model dynamics), $\vec{r}_p(t)$ is the state of the perturbed system, and $\vec{d}$ denotes external disturbance that could be either a stochastic noise, a $L^2$-bounded, even an $L^\infty$-bounded signal, e.g., a unit step, we assume:
\begin{subequations}
\begin{align}
\begin{split}
\tfrac{d}{dt}\vec{r}_u &= \A \vec{r}_u + {\bm B} \vec{d}, \qquad\qquad \vec{r}_u(0)=\vec{r}_{u,0},\\
y_u &= \C_u \vec{r}_u;
\end{split} \label{eq:r_unperturbed}\\[1ex]
\begin{split}
\tfrac{d}{dt}\vec{r}_p &= (\A+\delta \S) \vec{r}_p + {\bm B} \vec{d}, \quad \vec{r}_p(0)=\vec{r}_{p,0},\\
y_p &= \C_p \vec{r}_p.
\end{split}\label{eq:r_perturbed}
\end{align}
\end{subequations}
The multivariable controllable canonical form freezes the $\bm B$-matrix, so that the uncertainties are lumped into $\A$ and $\C$. The uncertainty in $\A$ is assumed to be structured with \emph{structure} $\S$ and \emph{magnitude} $\delta$, while the uncertainty on $\C$ is left unstructured for now. A slight difficulty with the utilization of a controllable canonical form is that the same uncertain parameter could appear in both $\A$ and $\C$, 
but this is easily dealt with as argued in Corollary~\ref{cor:SsameasSc}. State or output feedback to achieve specifications are assumed to be embedded in $\A$. In other words, Eqs.~\eqref{eq:r_unperturbed} and~\eqref{eq:r_perturbed} represent the \emph{closed-loop} unperturbed and perturbed systems, resp., designed by any means, including  learning~\cite{PRA} and competitive analysis~\cite{CompetitiveControl}. The output error $\vec{e}(t):=\C_p\vec{r}_p-\C_u\vec{r}_u$ admits two different, but equivalent, $\vec{d}$-agnostic state-space realizations, both of which have state $\vec{z}(t):=\vec{r}_p(t)-\vec{r}_u(t)$:
\begin{subequations}
\begin{align}\label{eq:z_unperturbed}
\begin{split}
\tfrac{d}{dt}\vec{z} &= (\A+\delta \S) \vec{z} + \delta \S \vec{w}_u, \qquad \vec{z}(0)=\vec{z}_0,\\
\vec{e} &=\C_p\vec{z} +(\C_p-\C_u)\vec{w}_u;
\end{split}\\[1ex]
\label{eq:z_perturbed}
\begin{split}
\tfrac{d}{dt}\vec{z} &= \A\vec{z} + \delta \S\vec{w}_p, \qquad \qquad\quad\; \vec{z}(0)=\vec{z}_0,\\
\vec{e} &=\C_u \vec{z}+(\C_p-\C_u)\vec{w}_p.
\end{split}
\end{align}
\end{subequations}
The difference between the two models is in the driving terms: $\vec{w}_u:=\vec{r}_u$ for the \emph{unperturbed} error dynamics~\eqref{eq:z_unperturbed} and $\vec{w}_p:=\vec{r}_p$ for the \emph{perturbed} error dynamics~\eqref{eq:z_perturbed}. This entails different sources of potential instabilities: in the unperturbed formulation of Eq.~\eqref{eq:z_unperturbed}, assuming $\vec{w}_u \in L^2$, the instability is in the free dynamics $\A+\delta \S$; while in the perturbed case of Eq.~\eqref{eq:z_perturbed}, the instability is in the input term $\vec{w}_p$, which need not be in $L^2$ unless the perturbed dynamics is kept stable.  

The state-space error dynamics gives the transfer matrices
\begin{subequations}
  \begin{align}
    \T^u_{\vec{e},\vec{w}_u}(s,\delta)&:= (\C_p-\C_u)+\C_p(s\I-\A-\delta \S)^{-1}\delta \S\label{eq:Tu_error},\\
    \T^p_{\vec{e},\vec{w}_p}(s,\delta)&:=(\C_p-\C_u)+\C_u(s\I-\A)^{-1}\delta \S,\label{eq:Tp_error}
  \end{align}
\end{subequations}
where we have adopted the universal notation of $\T_{\vec{y},\vec{x}}$ to denote the transfer matrix from $\hat{\vec{x}}$ to $\hat{\vec{y}}$, and $\hat{\vec{x}}(s)$ denotes the Laplace transform of $\vec{x}(t)$. 

Application of the matrix inversion lemma reveals that $\T^u$ is the transfer function for the error scaled relative to the unperturbed dynamics, while $\T^p$ is scaled relative to the perturbed dynamics, as easily seen from
\begin{subequations}\label{eq:freq_weighing}
  \begin{align}
    \T^u_{\vec{e},\vec{w}_u}(s,\delta) = & \left[\C_p(s\I-\A-\delta \S)^{-1}- \C_u(s\I-\A)^{-1}\right] \nonumber\\
    & \cdot\left[(s\I-\A)^{-1}\right]^{-1}\label{eq:Tu_error_scaled},\\
    \T^p_{\vec{e},\vec{w}_p}(s,\delta) = &\left[\C_p(s\I-\A-\delta \S)^{-1} -\C_u(s\I-\A)^{-1} \right] \nonumber\\
    &\cdot\left[(s\I-\A-\delta \S)^{-1}\right]^{-1}.\label{eq:Tp_error_scaled}
\end{align}
\end{subequations}

Moreover, Eqs.~\eqref{eq:Tu_error_scaled} and~\eqref{eq:Tp_error_scaled} reveal that the two models differ by a frequency correction factor:
\begin{equation}\label{eq:equate}
  \T^p_{\vec{e},\vec{w}_p}(s,\delta)
  = \T^u_{\vec{e},\vec{w}_u}(s,\delta)\left((s\I-\A)^{-1}(s\I-\A-\delta \S) \right).
\end{equation}

Expansion with respect to the unperturbed state (case a) makes sense as we are expressing the scaled deviation of the dynamics (as captured by the error transfer function) in terms of something we know. Expansion with respect to the perturbed state (case b), however, is also interesting in that we are expressing how much the dynamics of the system deviates from the presumed dynamics in terms of the \emph{actual} state of the system.  This is a more relevant formulation if we can probe the actual dynamics of the system experimentally. For a historical review of this unperturbed vs. perturbed scaling dichotomy see~\cite{Safonov_Laub_Hartmann}.

\begin{remark} A direct feedthrough $\D\bm d$ could be added to the right-hand sides of the $y_{u},y_p$ terms in Eqs.~\eqref{eq:r_unperturbed}-\eqref{eq:r_perturbed}. If $\D$ is not subject to uncertainties, subtracting Eq.~\eqref{eq:r_unperturbed} from Eq.~\eqref{eq:r_perturbed} yields the same Eqs.~\eqref{eq:z_unperturbed}-\eqref{eq:z_perturbed}. However, a departure $\D_p$ from its unperturbed value $\D_u$ creates an extra term $(\D_p-\D_u)\bm d$ in the right-hand sides of the $\bm e$-terms in Eqs.~\eqref{eq:z_unperturbed}-\eqref{eq:z_perturbed}. Since the two disturbances $\bm w_{u/p}$ and $\bm d$ are decoupled, their respective effects can be treated separately: the former follows the lines of the paper and the latter is trivial as a direct transmission. \\
\end{remark}

\subsection{Contribution}\label{s:contribution}

What distinguishes this approach from {\it ``classical"} robust performance, as represented by~\cite{Zhou, PackardDoyle, Young2001StructuredSV, Tempo2013UncertainLS, Safonov84, Safonov89, Safonov95} is that the former deals with Eqs.~\eqref{eq:r_unperturbed}--\eqref{eq:r_perturbed}, while we focus on Eqs.~\eqref{eq:z_unperturbed}--\eqref{eq:z_perturbed} and carry out robust performance on Eqs.~\eqref{eq:Tu_error_scaled} and \eqref{eq:Tp_error_scaled}, rather than Eq.~\eqref{eq:r_perturbed}.  Put another way, classical robust performance bounds $\|\T_{y_p,\vec{d}}(s,\delta)\|_\infty$ whereas we bound $\|\T^{u}_{y_p-y_u,\vec{w}_{u}}(s,\delta)\|_\infty$ or $\|\T^{p}_{y_p-y_u,\vec{w}_{p}}(s,\delta)\|_\infty$.   An advantage of using $\|\T^{u/p}_{y_p-y_u,\vec{w}_{u/p}}(s,\delta)\|_\infty$ in the unperturbed or perturbed ($u/p$) case is that it allows an objective assessment of the effect of the uncertainties, rather than their effect on the transmission of an arbitrary external disturbance $\vec{d}$.  Both the unperturbed and perturbed formulations are driven by \emph{physically objective} driving terms: the states of the unperturbed and the perturbed dynamics, resp.  If an additive disturbance$\vec{d}$ is physically justifiable, the same approach decouples the effect of the disturbance and allows performance to be assessed solely as a consequence of uncertain parameters.

For oscillatory systems, such as lightly damped space structures~\cite{TRW,symmetric_passive}, or weakly decoherent quantum systems~\cite{rings_QINP, chains_QINP}, the unperturbed error dynamics has the advantage that the frequency sweep can be limited to the known eigenfrequencies of $\A$. The perturbed error dynamics still requires the classical frequency sweep because the eigenfrequencies of $\A+\delta \S$ are imprecisely known, but
in this case the classical \muSSV-computations can be completely circumvented.

\section{Robust performance in the unperturbed formulation} \label{sec:unperturbed}

This section proceeds from Eq.~\eqref{eq:z_unperturbed} with an \emph{unperturbed} driving term $\vec{\hat{w}}_u$. We begin with the simplest case of zero initial state error and exact knowledge of $C$ with all plant parameter errors lumped into $\A$. We then introduce the most general case of the unperturbed dynamics including an uncertain initial state and uncertain sensor matrix. The cases of particular combinations of uncertainties in the initial state or sensor matrix are then deduced.

\subsection{From zero initial condition error to generalization}\label{sec:ewu}

Elementary matrix manipulation reveals that the transfer matrix $\T^u_{\vec{e},\vec{w}_u}$ in the formulation of Eq.~\eqref{eq:Tu_error} is obtained from
\begin{equation}\label{eq:Gewu}
  \begin{bmatrix} \v\\\e(s)\end{bmatrix}
    =\underbrace{\begin{bmatrix}
    \bPhi(s)^{-1}\S & \bPhi(s)^{-1}\S\\\C_p & \C_p-\C_u \end{bmatrix}}_{\G_{\vec{e},\vec{w}_u}(s)}
    \begin{bmatrix} \Eta \\ \w_u(s)\end{bmatrix},\;
\end{equation}
with $\bPhi(s)=s\I-\A$ and feedback $\Eta=(\delta \I) \v$ so that $\T^u_{\vec{e},\vec{w}_u}(s,\delta) = \mathcal{F}_u(\G_{\vec{e},\vec{w}_u}(s),\delta \vec{I})$, the upper linear fractional transformation $\G_{\vec{e},\vec{w_u}}(s)$ with the uncertainty $\delta \vec{I}$ \cite{Zhou}. With this feedback, we compute $\|\T^u_{\vec{e},\vec{w}_u}(s,\delta)\|$ via a \emph{fictitious} feedback $\w_u = \Delta_f \e$, with $\Delta_f$ a complex fully populated matrix. Specifically, a simple singular value argument shows that
\begin{equation*}
  \|\T^u_{\vec{e},\vec{w}_u}(s,\delta)\|=1/\min\{\|\Delta_f\|:
  \det(\I-\T^u_{\vec{e},\vec{w}_u}(s,\delta)\Delta_f)=0\}.
\end{equation*}
To compute $\max_\delta \|\T^u_{\vec{e},\vec{w}_u}(s,\delta)\|$, subject to closed-loop stability, the two feedbacks are combined as
\begin{equation}\label{eq:boldDeltaIIIA}
\begin{bmatrix}
\Eta\\ \w_u
\end{bmatrix}=
\underbrace{\begin{bmatrix}\delta \I & 0\\0 & \bDelta_f\end{bmatrix}}_{\bDelta (\delta,\Delta_f)}
\begin{bmatrix}
\v \\ \e
\end{bmatrix},
\end{equation}
where $\bDelta$ defines the \emph{structure} $\mathcal{D}$, i.e., the set of all block-diagonal matrices with the top left hand corner a real scalar matrix and the bottom left hand corner a fully populated complex matrix. It is readily verified that
\begin{equation}\label{eq:2factors}
\begin{split}
& \det(\I-\G_{\vec{e},\vec{w}_u}(s)\bDelta(\delta,\bDelta_f))=\\
& \quad \det(\I-\bPhi^{-1}\S\delta)\det(\I-\T^u_{\vec{e},\vec{w}_u}(s,\delta)\bDelta_f),
\end{split}
\end{equation}
where the first factor relates to the closed-loop stability, as $\det(\I-\bPhi^{-1}\S\delta) = \det(\bPhi^{-1})\det(s\I-\A-\S\delta)$. 
With the $\G_{\vec{e},\vec{w}_u}(s)$ matrix properly constructed, 
invoking the robust performance theorem~\cite[Th. 10.8]{Zhou} immediately yields
\begin{theorem}\label{th:IIIA}
If $\bPhi(s)$ is invertible in the right-half complex plane (RHP), then $\|\T^u_{\vec{e},\vec{w}_u}(s,\delta)\|\leq \mu_{\mathcal{D}}(\G_{\vec{e},\vec{w}_u}(s)), \forall$ $\delta < 1/\mu_{\mathcal{D}}(\G_{\vec{e},\vec{w}_u}(s))$, where
\begin{equation*}
  \mu_{\mathcal{D}}(\G_{\vec{e},\vec{w}_u}(s))=
  \left[\min_{\bDelta \in \mathcal{D}} \{\|\bDelta\| : \det(\I-\G_{\vec{e},\vec{w}_u}(s) \bDelta)=0 \}\right]^{-1}
\end{equation*}
is the structured singular value specific to $\mathcal{D}=\left\{\bDelta (\delta,\Delta_f) \right\}$.$\blacksquare$
\end{theorem}

The above is a generic result that can easily be extended to more complicated uncertainty patterns. Given a transfer matrix $\T_{\vec{e},\vec{w}_g}(s,\delta_g)$ from a generalized disturbance to some error, subject to a generalized structured uncertainty of magnitude $\delta_g$, the difficulty is to find a matrix $\G_g$ and a structured feedback $\bDelta_g$, such that wrapping the structured feedback $\Eta_g=\bDelta_g \v_g$ around
\[
\begin{bmatrix} \v_g \\ \e \end{bmatrix} = \G_g \begin{bmatrix}\Eta_g\\ \w_g \end{bmatrix},
\]
reproduces $\T_{\vec{e},\vec{w}_g}(s,\delta_g) = \mathcal{F}_{u}(\G_{\vec{g}(s)},\vec{\Delta_g})$. In the subsections that follow, we derive $\G_g$ and $\bDelta_g$ and bound $\|\T_{\vec{e},\vec{w}_g}(s,\delta_g)\|$ without repeating the argument of Th.~\ref{th:IIIA} with the objective of bounding $\T_{\vec{e},\vec{w}_g}(s,\delta_g)$ with
\begin{equation*}
  \mu_{\mathcal{D}}(\G_g(s)) = \left[\min_{\bDelta \in \mathcal{D}}\{\|\bDelta\|: \det(\I-\G_g\bDelta)=0\} \right]^{-1},
\end{equation*}
where $\mathcal{D}=\{\mathrm{block}\mbox{-}\mathrm{diag}(\bDelta_g,\bDelta_f)\}$.

\subsection{Effect of initial state preparation error with uncertain $\C$}\label{ss: uncertain_C_and_z(0)}
With an understanding of how to derive $\G_{g}$ and the associated feedback structure in the simple case of Section~\ref{sec:ewu}, we establish the general case of unperturbed dynamics with an uncertain $C$ matrix and initial state error. Model the uncertainty in $C$ as 
 \begin{equation}
 \C_p = \C_u + \delta_c \vec{S_c}.
 \end{equation}
 Then taking $z(0) \neq 0$, ~\eqref{eq:z_unperturbed} yields
 \begin{align*}
     &\T_{\vec{e},(z(0),\vec{w_u})} = \\& \begin{bmatrix} \left( \vec{C_u} + \delta_c \vec{I_c} \right) \left(\vec{I} - \vec{\Phi^{-1}} \vec{S} \delta \right)^{-1}\vec{\Phi^{-1}} \nonumber\\ \left( \left( \vec{C_u} + \delta_c \vec{I_c} \right) \left( \vec{I} - \vec{\Phi^{-1}}\vec{S} \delta \right)^{-1}\vec{\Phi^{-1}} \vec{S} \delta + \delta_c \vec{S_{c}} \right)\end{bmatrix}^T
 \end{align*}
 with input $\left[ z(0) \ \vec{w_u} \right]^{T}$. While the structure of $\G^c_{\vec{e},(\vec{z}(0),\vec{w}_u)}$ is not readily apparent from this equation, we can ``pull out'' the uncertainties $\delta$ and $\delta_c$ as depicted in Figure~\ref{fig:blockdiagram}.
 \begin{figure}
    \centering
    \includegraphics[trim=20 35 350 30, clip,width=\columnwidth]{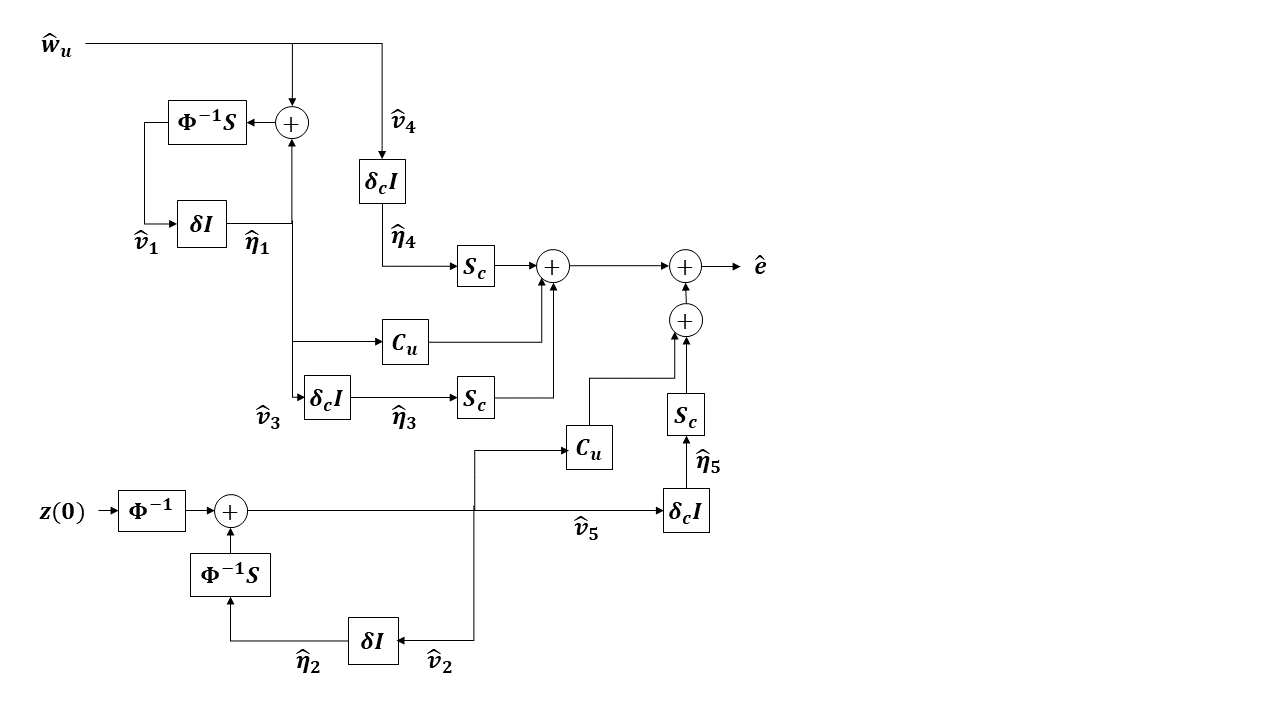}
    \caption{Diagram for error dynamics driven by unperturbed state in the most general case: with initial state preparation error and uncertain $C$.}
    \label{fig:blockdiagram}
\end{figure}
Associating the system inputs $\vec{\hat{\eta}}_n$ from each uncertainty block  $\delta \vec{I}$ or $\delta_c \vec{I_c}$ with the output $\vec{\hat{v}}_n$ results in $\vec{\hat{\eta}}=\vec{\Delta_g} \vec{\hat{v}}$ with 
\begin{equation*}
  \vec{\Delta_g}
  =
  \left[\begin{array}{ccccc}
    \delta \I & 0 & 0 & 0 & 0 \\
    0 & \delta \I & 0 & 0 & 0\\
    0 & 0 & \delta_c \I & 0 & 0\\
    0 & 0 & 0 & \delta_c \I & 0\\
    0 & 0 & 0 & 0 & \delta_c \I
  \end{array}\right].
\end{equation*}
Writing each output $\hat{\vec{v}}_n$ and $\vec{\hat{e}}$ in terms of the inputs $\vec{\hat{\eta}}_n$, $\z(0)$ and $\vec{\hat{w}}_u$ yields
\begin{align}\label{eq:Gez0wuC}
& \left[ \begin{array}{c} \vec{\hat{v}} \\ \hline \vec{\hat{e}} \end{array} \right] = \left[ \begin{array}{c|c} \vec{G_{11}} & \vec{G_{12}} \\ \hline \vec{G_{21}} & \vec{G_{22}}\end{array} \right] \left[ \begin{array}{c} \vec{\hat{\eta}} \\ \hline z(0) \\ \vec{\hat{w}}_u \end{array} \right] = \\&
  \begin{bmatrix}
    \v_1 \\ \v_2 \\ \v_3 \\ \v_4 \\ \v_5 \\ \cmidrule(lr){1-1} \e \end{bmatrix}
   \! = \! \underbrace{\left[\begin{array}{ccccc|cc}
    \X & 0 & 0 & 0 & 0 & 0& \X\\
    0 & \X &0&0&0&\bPhi^{-1} & 0 \\
    \I & 0 & 0 & 0 & 0 & 0 & 0\\
    0 & 0 & 0 & 0 & 0 & 0 &\I\\
    0 & \X & 0 & 0 & 0 & \bPhi^{-1} & 0\\\hline
    \C_u & \C_u \X & \S_c & \S_c & \S_c & \C_u\bPhi^{-1} & 0
  \end{array}\right]}_{G^c_{\vec{e},(\vec{z}(0),\vec{w}_u)}} \!
  \begin{bmatrix}
    \Eta_1 \\ \Eta_2 \\ \Eta_3 \\ \Eta_4 \\ \Eta_5 \\ \cmidrule(lr){1-1} \vec{z}(0) \\ \w_u
  \end{bmatrix} \notag
\end{align}
from which $\G^{c}_{\vec{e},(z(0),\vec{w_u})}$ is readily identified. Here $\vec{X} = \vec{\Phi}^{-1}\vec{S}$. Proceeding as in Section~\ref{sec:ewu}, we close the loop from $\vec{\hat{e}}$ to $\left[ z(0) \ \vec{\hat{w}}_u \right]^T$ with the full, fictitious feedback $\vec{\Delta}_f$ to produce the overall feedback $\bDelta(\delta, \delta, \delta_c, \delta_c, \delta_c, \bDelta_f)$. 
Consequently, proceeding 
as in Theorem~\ref{th:IIIA} and by the robust performance theorem~\cite[Th. 10.8]{Zhou} we have
\begin{theorem}\label{th:IIID} 
If $\Phi(s)$ is invertible in the RHP,  setting $\mathcal{D}=\left\{\bDelta(\delta, \delta, \delta_c, \delta_c, \delta_c, \bDelta_f) \right\}$, we have
  \begin{equation*}
    \|\T^u_{\vec{e},(z(0),\vec{w}_u)}(s,\delta,\delta_c)\|
    \leq\mu_{\mathcal{D}}(\G_{\vec{e},(z(0),\vec{w}_u)}^c(s))
  \end{equation*}
  for $\delta,\delta_c<1/\mu_{\mathcal{D}}(\G_{\vec{e},(z(0),\vec{w}_u)}^c(s))$. $\blacksquare$
\end{theorem}

\begin{corollary}\label{cor:SsameasSc}
If the same uncertainy amplitude parameter $\delta$ appears in both $A$ and $C$ with uncertainty structure $S$ and $S_c$, resp.,  
Th.~\ref{th:IIIA} remains valid after equating $\delta$ and $\delta_c$ resulting in the overall feedback is $\vec{\Delta}(\delta,\delta,\delta,\delta,\delta,\vec{\Delta}_f)$. $\blacksquare$
\end{corollary}

In the special case where we have an uncertain $C$ matrix with no initial state error, $z(0) = 0$, the inputs $\vec{\hat{\eta}_2}$ and $\vec{\hat{\eta}_5}$ are both zero and we can simplify Eq.~\eqref{eq:Gez0wuC} by deleting the second and fifth row and column of $\G_{\vec{e},(z(0),\vec{w}_u)}^c$.
Similarly, if we have no uncertainty in the sensor matrix $C$ but an initial state error then $ 0 = \vec{C}_p - \vec{C}_u$ allows the reduction $\vec{\hat{\eta}} = \left[ \vec{\hat{\eta}_1} \ \vec{\hat{\eta}_2} \right]^T$ and $\vec{\hat{v}} = \left[ \vec{\hat{v}_1} \ \vec{\hat{v}_2} \right]^T$ and again we can simplify Eq.~\eqref{eq:Gez0wuC}.

\section{Robust performance in the perturbed formulation}\label{sec:perturbed}

We consider the point of view taken by Eq.~\eqref{eq:z_perturbed} with a \emph{perturbed} driving term $\vec{w}_p$. Taking its Laplace transform yields
\begin{subequations}\label{eq:L_z_perturbed}
  \begin{align}
    (s\I-\A) \z(s) &= \delta\S \w_p(s) + \vec{z}(0),\\
    \e(s)          &= \C_u \z(s) + (\C_p-\C_u) \w_p(s).
  \end{align}
\end{subequations}

\subsection{Zero initial error}\label{sec:ewp}

If $s\I-\A$ is invertible and there is no initial state preparation error, introducing the structured uncertainty $\C_p-\C_u=\delta_c\S_c$ yields
\begin{equation}
  \T_{\vec{e},\vec{w}_p}^p(s,\delta,\delta_c) = \C_u \bPhi^{-1}(s)\delta \S+\delta_c \S_c.
\end{equation}
Bounding $\|\T_{\vec{e},\vec{w}_p}^p(s,\delta)\|$ can be trivially although conservatively done by
\begin{equation}
  \|\T_{\vec{e}\vec{w}_p}^p(s,\delta,\delta_c)\| \leq
  \| \begin{bmatrix} \C_u\bPhi^{-1}(s)\S & \S_c \end{bmatrix}\| \sqrt{\delta^2+\delta^2_c},
\end{equation}
without recourse to the robust performance theorem. 
The perturbed case, hence, offers a new approach, circumventing the \muSSV~analysis.

It is, however, of interest to approach the problem in the robust performance context, observing that $\T_{\vec{e},\vec{w}_p}^p(s,\delta,\delta_c)$ is obtained from
\begin{equation}\label{eq:GewpC}
  \begin{bmatrix} \v_1 \\ \v_2 \\ \cmidrule(lr){1-1} \e \end{bmatrix} =
  \underbrace{\left[\begin{array}{cc|c}
    0 & 0 & \I\\
    0 & 0 & \I\\\hline
    \C_u\bPhi^{-1}\S & \S_c & 0
  \end{array}\right]}_{G^c_{\vec{e},\vec{w}_p}}
  \begin{bmatrix}
    \Eta_1 \\ \Eta_2 \\ \cmidrule(lr){1-1} \w_p
  \end{bmatrix}
\end{equation}
after the feedback
\begin{equation}\label{eq:feedback}
  \begin{bmatrix} \Eta_1 \\ \Eta_2\end{bmatrix}
  =\begin{bmatrix} \delta \I & 0 \\ 0 & \delta_c \I\end{bmatrix}
  \begin{bmatrix} \v_1 \\ \v_2\end{bmatrix}.
\end{equation}
With $\bDelta(\delta,\delta_c,\bDelta_f)$ as defining the structure $\mathcal{D}$, we have the following theorem, which differs from the classical robust performance of Th.~\ref{th:IIIA}:
\begin{theorem}\label{th:IVA}
If $\Phi(s)$ is invertible in the RHP,  $\|\T^p_{\vec{e},\vec{w}_p}(j\omega,\delta,\delta_c)\|\leq \mu_{\mathcal{D}}(\G_{\vec{e},\vec{w}_p}^c(j\omega))$, for $\delta$ such that $\vec{w}_p \in L^2$.
\end{theorem}
\begin{proof}
Observe that
$\det(\I-\G_{\vec{e},\vec{w}_u}^c(s)\bm\Delta)=\det(\I-\T^p_{\vec{e},\vec{w}_u}(s,\delta,\delta_c)\bDelta_f)$. Enforcing $\det(\I-\G_{\vec{e},\vec{w}_u}^c(\jmath \omega)\bm\Delta)\ne 0$ therefore bounds
$\T^p_{\vec{e},\vec{w}_u}(j\omega,\delta,\delta_c)$ but does not enforce robust stability which is dealt with separately by enforcing $\vec{w}_p$ to be a square integrable signal.
\end{proof}

Observe that the bound of Th.~\ref{th:IVA} holds for all $\delta$'s, but that $\|\T^p_{\vec{e},\vec{w}_p}(j\omega,\delta,\delta_c)\|_\infty$ is a valid operator norm only for $w_p\in L^2$, hence the restriction on $\delta$.

\subsection{Effect of initial state preparation error}\label{sec:pz0S}

A significant difference between the unperturbed and the perturbed case is that in the latter the effect of $\vec{z}(0)$ is completely decoupled from $\vec{\hat{w}}_p$, as seen by the transfer matrix
\begin{equation}
  \T_{\vec{e},(\vec{z}(0),\vec{w}_p)}^p(s,\delta,\delta_c)
  = \begin{bmatrix}
    \C_u \bPhi^{-1}(s) & \C_u \bPhi^{-1}(s)\delta \S+\delta_c \S_c
  \end{bmatrix}.
\end{equation}
Bounding $\|\T_{\vec{e},(z(0),\vec{w}_p)}^p(s,\delta,\delta_c)\|$ is again trivial,
\begin{equation} \label{eq:Tbound}
  \begin{split}
    \|\T_{\vec{e},(\vec{z}(0),\vec{w}_p)}^p(s,\delta,\delta_c)\|&\leq  \|\C_u\bPhi^{-1}(s)\|\\
    &+ \| \begin{bmatrix}\C_u \bPhi^{-1}(s) \S& \S_c \end{bmatrix}\| \sqrt{\delta^2+\delta_c^2},
  \end{split}
\end{equation}
and does not require the robust performance analysis. It, hence, offers a new approach.

If a robust performance analysis based on structured singular value analysis is desired, it can be accomplished via
\begin{equation}\label{eq:Gez0wpC}
  \begin{bmatrix}
    \v_1\\ \v_2\\\cmidrule(lr){1-1} \e
  \end{bmatrix} =
  \underbrace{\left[\begin{array}{cc|cc}
    0 & 0 & 0 & \I\\
    0 & 0 & 0 & \I\\\hline
    \C_u\bPhi^{-1}\S & \S_c & \C_u\bPhi^{-1} & 0
  \end{array}\right]}_{G^c_{\vec{e},(\vec{z}(0),\vec{w}_p)}}
  \begin{bmatrix}
    \Eta_1 \\ \Eta_2\\\cmidrule(lr){1-1} \vec{z}(0)\\ \w_p
  \end{bmatrix}
\end{equation}
after the same feedback as Eq.~\eqref{eq:feedback}. With the $\mathcal{D}$ structure defined by $\bDelta(\delta,\delta_c,\bDelta_f)$, we have
\begin{theorem}\label{th:IVB}
If $\Phi(s)$ is invertible in the RHP,  $\|\T^p_{\vec{e},(z(0),\vec{w}_u)}(s,\delta,\delta_c)\|\leq
  \mu_{\mathcal{D}}(\G_{\vec{e},(z(0),\vec{\hat{w}}_u)}^c(s))$ for $s=j\omega$ and  $\delta$ such that $\vec{w}_p \in L^2$. $\blacksquare$
\end{theorem}

\section{Benchmark mechanical example: Double spring-mass-dashpot}\label{sec:mechanical_example}

\begin{figure*}
\subfloat[Uncertain stiffness $k_1$\label{fig:k1}]
{\includegraphics[width=0.49\textwidth,height=2in]{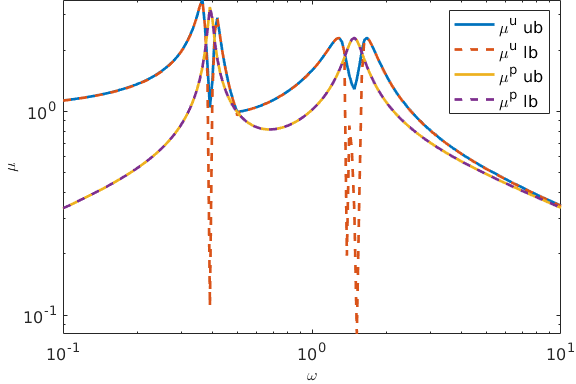}}
\hfill
\subfloat[Uncertain stiffness $k_2$\label{fig:k2}]{
\includegraphics[width=0.49\textwidth,height=2in]{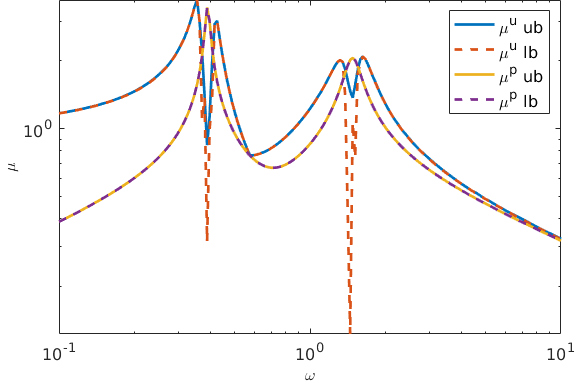}}
\\
\subfloat[Uncertain damping $b_1$ \label{fig:b1}] {\includegraphics[width=0.49\textwidth,height=2in]{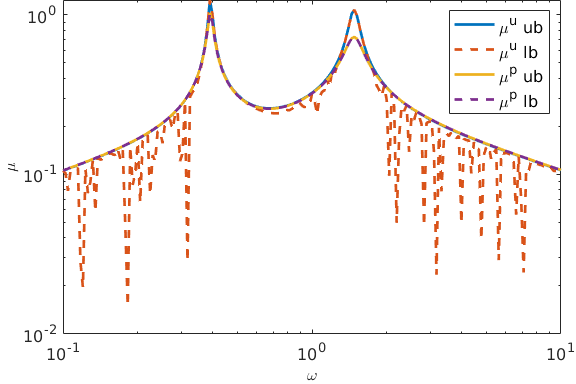}}
\hfill
\subfloat[Stiffness $k_1$ uncertainty and sensor-actuator misalignment ($S_c$) \label{fig:SSc}] {\includegraphics[width=0.49\textwidth,height=2in]{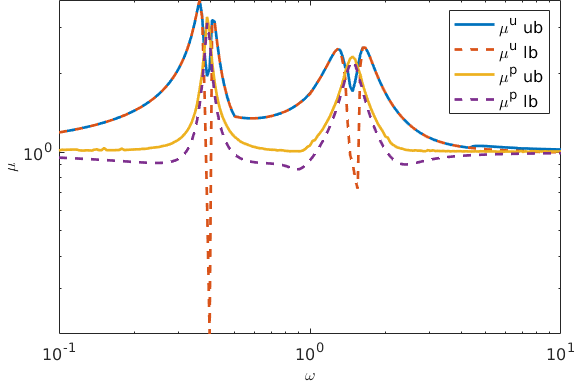}}
\\
\subfloat[Combination of stiffness $k_1$ uncertainty and initial error $\vec{z}(0)$\label{fig:k1z0}] {\includegraphics[width=0.49\textwidth,height=2in]{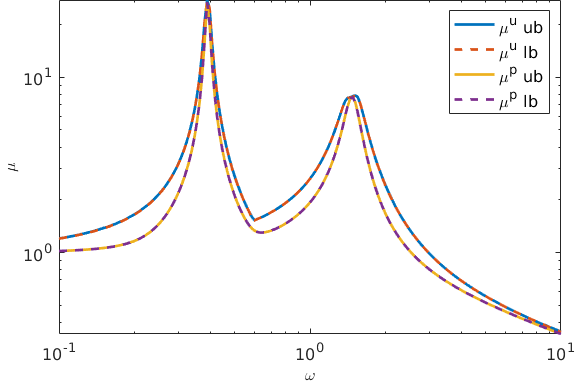}}
\hfill
\subfloat[Stiffness $k_1$ uncertainty, initial error $\vec{z}(0)$ and $S_c$ uncertainty $C$\label{fig:k1z0Sc}]{\includegraphics[width=0.49\textwidth,height=2in]{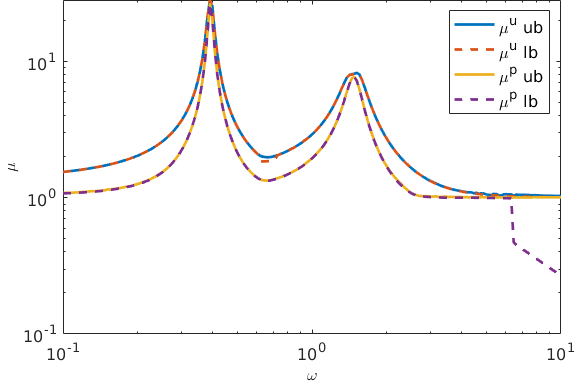}}
\caption{Upper and lower bounds on the structured singular value $\mu$ for unperturbed and perturbed formulations for different combinations of uncertainties for the spring-mass dashpot system with $m_1=3, m_2=1$, $k_1=k_2=1$, $b_1=b_2=0.1$.}
\end{figure*}

To illustrate the perturbed vs. unperturbed approach, we consider a simple system that has been a benchmark problem for robust control synthesis~\cite{Benchmark}: the double spring-mass-dashpot system of~\cite[Example 4.2]{Zhou}, where $d_1$, $d_2$ are the forces acting on masses $m_1$, $m_2$, resp., and $e_1$, $e_2$ are the rates of motion of masses $m_1$, $m_2$, resp., relative to the spring rest positions:
\begin{subequations}\label{eq:benchmark}
  \begin{align}
  \frac{d}{dt}\vec{r}_u &=
    \begin{bmatrix} \Z & \I\\ \M_1 & \M_2 \end{bmatrix} \vec{r}_u
    +\begin{bmatrix} \Z \\ \I\end{bmatrix}\vec{d},\\
  \vec{e} &=\left[\begin{array}{cc|cc}
      0 & 0 & m_1^{-1} & 0\\
      0 & 0 & 0 & m_2^{-1}
    \end{array}\right]\vec{r}_u,
  \end{align}
\end{subequations}
where $\Z = \begin{bmatrix} 0 & 0 \\ 0 & 0 \end{bmatrix}$, $\I= \begin{bmatrix} 1 & 0 \\ 0 & 1 \end{bmatrix}$ is the identity matrix, and
\begin{gather*}
  \M_1 = \begin{bmatrix}
    -\frac{k_1}{m_1} & \frac{k_1}{m_2} \\
    \frac{k_1}{m_1} & -\frac{k_1+k_2}{m_2}
  \end{bmatrix}, \quad
  \M_2 =
  \begin{bmatrix}
    -\frac{b_1}{m_1} & \frac{b_1}{m_2}\\
    \frac{b_1}{m_1} & -\frac{b_1+b_2}{m_2}
  \end{bmatrix}.
\end{gather*}
$k_1, k_2>0$ and $b_1,b_2>0$ are the stiffness and damping constants of the springs, where the first spring connects $m_1$ to a fixed point and the second spring connects $m_1$ to $m_2$. As $m_1$ and $m_2$ appear in both $\A$ and $\C$, if the masses are uncertain, a minor modification of the method of Sec.~\ref{ss: uncertain_C_and_z(0)} is needed, 
as formalized in Corollary~\ref{cor:SsameasSc}.

\subsection{Positive real design}\label{s:positive_real}

In this case-study, $\T_{\vec{e},\vec{d}}$ is a mapping from force actuators to \emph{co-located} rate sensors, well known to be \emph{passive}, \emph{dissipative} or \emph{positive real}~\cite{positive_real, symmetric_passive}. Here, co-location means that the force $d_i$ applies to the mass $m_i$ and that $e_i$ measures the rate of motion of $m_i$. Passive means that, for $\vec{r}_u(0)=0$, there exists a constant $\epsilon>0$ such that $\int_0^T \vec{e}^T(t) \vec{d}(t)\,dt > \epsilon\|\vec{d}\|^2_{L^2[0,T]}$, $\forall T>0, \forall \vec{d}\in L^2[0,T]$. For such a system, any feedback $\vec{d}=-D\vec{e}=-DC\vec{r}_u$, where $D=D^T>0$, preserves closed-loop stability, similar to a bias-field-mediated feedback in quantum spintronic systems, which also preserves stability~\cite{rings_QINP}.

This ``positive real'' design~\cite{positive_real,symmetric_passive} has been widely applied to space structures with the caveat that for distributed parameter systems with vibration eigenfrequencies $\omega_{k \to \infty} \to \infty$, it is difficult to maintain co-location of the sensors and actuators. Indeed, if $\varepsilon$ is the distance between the actuating and the sensing points, there exists a $k$ such that $\varepsilon \in [\lambda_k/2, \lambda_k]$, where $\lambda_k=2 \pi c/\omega_k$ is the wavelength of the vibration mode $\omega_k$ and $c$ is the propagation speed of the elastic disruption. In this situation, the feedback from the sensing point to the actuating point is no longer negative, but positive. Hence, it destabilizes the vibration mode $\omega_k$.  In the setup of Eq.~\eqref{eq:benchmark}, a co-location error can be modeled as the structured perturbation
\begin{equation}\label{eq:Sc}
  \S_c=\left[\begin{array}{cc|cc}
    0 & 0 & 0 & m_2^{-1}\\
    0 & 0 & m_1^{-1} & 0
  \end{array}\right],
\end{equation}
meaning that $e_1$ blends the rates of motion of $m_1$ and $m_2$.

Positive realness also has a circuit theory formulation, where $\vec{d}$ is the voltage and $\vec{e}$ the current of a passive $n$-port~\cite{Belevitch}. Such circuit models have been used to simulate quantum phenomena~\cite{Ezawa1}, in particular four-dimensional topological insulators~\cite{topological_insulators}.  Co-location in this context means that the voltage is measured at the same port at which the current is injected.

\subsection{Uncertain stiffness or damping}

We first consider a relative uncertainty in the stiffness $k_1$ modeled as $k_1(1+\delta)$. The unperturbed $\mu^u$ is computed using Th.~\ref{th:IIIA}; the perturbed $\mu^p$ is computed using Th.~\ref{th:IVA} after removal of the second row and column of $\G_{\vec{e},\vec{w}_p}^c$ as defined by Eq.~\eqref{eq:GewpC}, as $\delta_c=0$.  The results, illustrated in Fig.~\ref{fig:k1}, reveal two regimes: $\mu^p(\jmath \omega) > \mu^u(\jmath \omega)$ around the resonant frequencies and the reverse inequality $\mu^p(\jmath \omega) < \mu^u(\jmath \omega)$ away from the resonances.  To explain this phenomenon, consider
\begin{subequations}\label{eq:Gu_vs_Gp}
  \begin{align}
    \det(\I-\G^u\bm\Delta) =& \det(\I-\Phi^{-1}\delta \S)\times\nonumber\\
    & \det(\I-\T^p(I-\Phi^{-1}\delta S)^{-1}\bDelta_f),\\
    \det(\I-\G^p\bm\Delta) =& \det(\I-\T^p\bDelta_f),
  \end{align}
\end{subequations}
which is obtained by direct matrix manipulations and replacing $\T^u$ by its expression derived from Eq.~\eqref{eq:equate}. At the exact unperturbed resonance, $\|\T^p\|$ reaches its maximum, while $\T^u=\T^p(\I-\bPhi^{-1}\delta \S)^{-1}$ is damped by $(\I-\bPhi^{-1}\delta \S)^{-1}$ since $\bPhi^{-1}$ also achieves its maximum at the resonant eigefrequencies. Therefore, at exact resonance, $\|\T^u\| < \|\T^p\|$. Denoting by $\bDelta_f^u$, $\bDelta_f^p$ the $\bDelta_f$, for which the determinant in Eq.~(\ref{eq:Gu_vs_Gp}a) and Eq.~(\ref{eq:Gu_vs_Gp}b), resp., vanishes, we have $\|\bDelta_f^u\|>\|\bDelta_f^p\|$, while $|\delta|$, the relative error on the stiffness, is bounded by $1$ to preserve closed-loop stability. Hence $\mu^u=1/\|\bDelta_f^u\|<\mu^p=1/\|\bDelta_f^p\|$, as shown by Figs.~\ref{fig:k1}, \ref{fig:k2} at the exact resonant frequency. Away from both the exact and the perturbed resonances, and under the assumption that both $\|\T^u\|$ and $\|\T^p\|$ are sufficiently damped (precisely $<1$), $\delta$ becomes the deciding factor of $\mu$ in the sense that $\mu^u=1/\min_{\bDelta}\max\{\|\bDelta_f\|, |\delta|\}=|\delta|=1$ for the unperturbed case. For the perturbed case, since $\delta$ is not involved in stability, $\mu^p=1/\min_{\bDelta}\max\{\|\bDelta_f\|, |\delta|\}=\|\T^p\|<1$. Hence, under sufficient damping far from the resonance, $\mu^u>\mu^p$, again shown by Figs.~\ref{fig:k1}, \ref{fig:k2}.

With regard to the damping parameters, the linear damping model is an approximation of an essentially nonlinear stress-strain hysteresis in the spring material.  The results of uncertain damping $b_1(1+\delta)$ for the same numerical example in Eq.~\eqref{eq:benchmark} using Theorems~\ref{th:IIIA},~\ref{th:IVA} are shown in Fig.~\ref{fig:b1}.  Observe that the upper bounds on the unperturbed and perturbed $\mu$ reach their maxima at the same frequency $\bar{\mu}^p(0.79)=\bar{\mu}^u(0.79)=0.65$ and $\bar{\mu}^p(2.2)=\bar{\mu}^u(2.2)=1.3$. This is not surprising, as uncertainty in the damping does not change the resonant frequency.

\subsection{Sensor-actuator misalignment and initial state errors}

In the presence of sensor mis-alignment, $\mu$ is computed by appealing to Th.~\ref{th:IIID}, Eq.~\eqref{eq:Gez0wuC} with $z(0)=0$ and the second and fifth columns and rows of 
$\G^{c}_{\vec{e},(z(0),\vec{w_u})}$ removed in the unperturbed case, and Th.~\ref{th:IVA}, Eq.~\eqref{eq:GewpC} in the perturbed case.  The results are shown in Fig.~\ref{fig:SSc}.  As before the unperturbed and perturbed cases yield close results at the resonant frequencies that, here, are shifted relative to the stiffness $k_1$ uncertainty case.

Keeping $k_1$ uncertain but adding an initial state error $\vec{z}(0)$ instead of a collocation error, as in Th.~\ref{th:IIID} of Sec.~\ref{ss: uncertain_C_and_z(0)} with $0 = \vec{C}_p - \vec{C}_u$ for the unperturbed case and Th.~\ref{th:IVB} of Sec.~\ref{sec:pz0S} for the perturbed case, the unperturbed $\mu$ is computed from Eq.~\eqref{eq:Gez0wuC} after setting $\C_p=\C$ and $\C_p-\C_u=0$ and deleting rows and columns 3 through 5 of $\G^{c}_{\vec{e},(z(0),\vec{w_u})}$.  The perturbed $\mu$ is computed from Eq.~\eqref{eq:Gez0wpC}, after removal of the second row and second column of $\G^c_{\vec{e},(\vec{z}(0),\vec{w}_p)}$. The results in Fig.~\ref{fig:k1z0} show again consistency between the unperturbed and perturbed cases. It is also worth noting that the $\mu$ has increased compared with the simplest case of $k_1$ uncertainty, as a result of the initial state error.

The combined effect of uncertainties in $k_1$, $z(0)$ and $\S_c$ (for the same model system) are shown in Fig.~\ref{fig:k1z0Sc}.  The unperturbed case is dealt with by Th.~\ref{th:IIID} and the perturbed case by Th.~\ref{th:IVB}.  

\subsection{Comparison of \muSSV and direct bounds}

In the perturbed case Eq.~\eqref{eq:Tbound} provides explicit bounds on the norm of the transfer function, which we can compare to the bounds obtained from $\mu$-analysis.    Table~\ref{tab:Bounds:compare} shows that the explicit bounds on 
$\|\T_{\vec{e},(\vec{z}(0),\vec{w}_p)}^p(s,\delta,\delta_c)\|$
for $\delta=1/\mu^p$ in the absence of initial state and co-location errors are in excellent agreement with the upper bounds obtained for $\mu^p$ for different structured perturbations.  
\begin{table} \centering
\begin{tabular}{|l|c|c|c|c|}
\hline
& $\omega$ & $\mu^p(i\omega)$ & $\delta_{\max}(s)=\mu^{-1}(s)$ & $\|\T^p(s,\delta_{\max},0)\|$ \\\hline
$S_1$ & 0.389  & 3.2502	& 0.3077 & 3.2502 \\
      & 1.4791 & 2.2981	& 0.4351 & 2.2981 \\\hline
$S_2$ &	0.389  & 3.4448	& 0.2903 & 3.4448 \\
      & 1.4791 & 2.0611	& 0.4852 & 2.0611 \\\hline
$S_3$ & 0.389  & 1.0278	& 0.973	 & 1.0278 \\
      & 1.4791 & 0.7267	& 1.276	 & 0.7267 \\\hline
$S_4$ & 0.3890 & 1.0893	& 0.918	 & 1.0893 \\
      & 1.4791 & 0.6518	& 1.5340 & 0.6518 \\\hline
\end{tabular}
\caption{Comparison of $\mu(s)$ for the perturbed case vs $\|\T_{\vec{e},(\vec{z}(0),\vec{w}_p)}^p(s,\delta,\delta_c)\|$
calculated according to Eq.~\eqref{eq:Tbound} with $\delta=\mu(s)^{-1}$, $\delta_c=0$, $s=i\omega$ for structured perturbations $S_1$ to $S_4$ corresponding to perturbation of the stiffness constants $k_1$, $k_2$ and damping rates $b_1$, $b_2$, respectively.}
\label{tab:Bounds:compare}
\end{table}

\section{Application to controlling quantum systems}\label{sec:quantum_examples}

Another enlightening application is control of quantum systems. Coherent quantum systems are open-loop purely oscillatory, and neither coherent open-loop controls nor physically relevant structured uncertainties like $J$-coupling errors can change the closed-loop oscillatory situation~\cite{rings_QINP, chains_QINP}. For open quantum systems, decoherence acts as a stabilizing controller~\cite{CDC_decoherence}, but there are still challenges in applying conventional structured singular value analysis. For example, constants of motion, such as the unit-trace constraint for density operators describing quantum states, create a pole at $0$ in $\bPhi^{-1}(s)$ in the real Bloch representation of the dynamics commonly used to describe open quantum systems. However, as noted in Sec~\ref{sec:error-dynamics}, since the unperturbed error dynamics is driven by the unperturbed dynamics, the frequency sweep could be limited to the resonant frequencies of $\A$, which are by definition nonvanishing; hence the (possibly multiple) closed-loop pole at $0$ is avoided. The perturbed formulation is theoretically more challenging, since the frequency sweep should include the uncertain resonant frequencies of $\A+\delta \S$, and passing over $\omega=0$ cannot be ruled out. If, while numerically exploring $\mu_\mathcal{D}(\omega \to 0)$, it appears that the maximum could be at $\omega=0$, with the difficulty that $\bPhi^{-1}(0)$ fails to exist, then the formal procedure developed in~\cite{robust_performance_open} should be followed. Essentially, $\bPhi^{-1}(s)$ should be replaced by $\bPhi^{\#}(s)$, where $\#$ denotes a specialized pseudo-inverse close in spirit to, but different from, the Moore-Penrose pseudo-inverse. This specialized pseudo-inverse applied to the construction of $\mu_\mathcal{D}$ cures the lack of continuity at $s=0$ and moreover for $s \ne 0$ coincides with the $\mu_\mathcal{D}$ derived from $\bPhi^{-1}(s)$. Note that in~\cite{soneil_mu}, this case has been approached heuristically by damping the closed-loop systems as $\A-\epsilon \I$ and then allowing $\epsilon \downarrow 0$.

\subsection{Bias field control of coupled qubit Rabi oscillations}

\begin{figure*}
\subfloat[Transfer fidelity oscillations]{\includegraphics[width=0.49\textwidth]{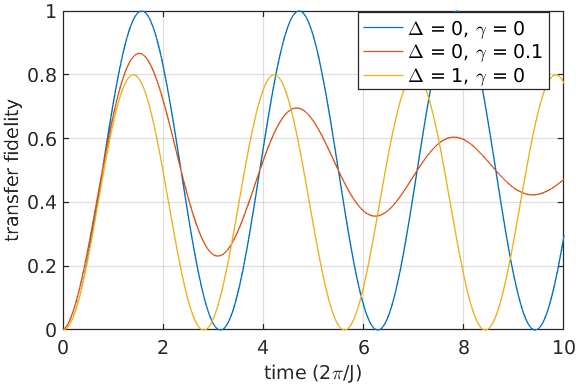}} \hfill
\subfloat[Transfer fidelity and time vs detuning]{\includegraphics[width=0.49\textwidth]{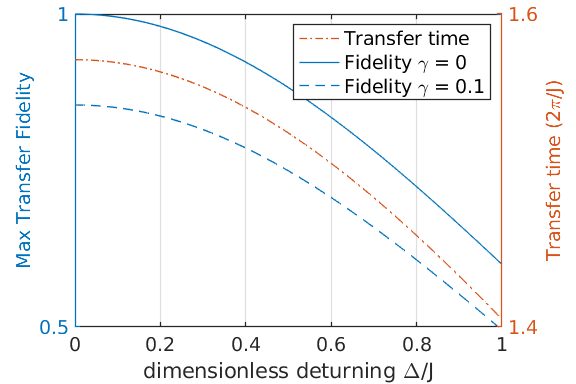}}
\caption{Transfer fidelity vs. time and maximum transfer fidelity and minimum transfer time as a function of detuning $\bDelta$ for quantum example.}\label{fig:Quantum1}
\end{figure*}

As a concrete example, consider a simple system of two qubits under XX-coupling of strength $J$, restricted to the single excitation subspace spanned by the states $\ket{L} = \ket{\uparrow \downarrow}$ and $\ket{R} = \ket{\downarrow \uparrow}$, which correspond to excitation of the left and right qubit, respectively. It is instructive to consider the two-qubit problem, which is analytically solvable, but the model can be extended to chains of more than two qubits and more complex networks. While the latter networks are generally not amenable to analytic solutions, numerical optimization suggests that high fidelities and good transfer times can be achieved by non-zero optimal biases, and that transfer fidelities may be robust to uncertainties in system parameters under certain conditions~\cite{statistical_control, Edmond_IEEE_AC}.

\subsection{Dynamics}

For two qubits the coherent dynamics in the single-excitation subspace can be modeled by the Hamiltonian
\begin{equation*}
  H = \hbar \begin{bmatrix} D_1 & J \\J & D_2 \end{bmatrix},
\end{equation*}
where $\hbar$ is the reduced Planck constant and $\hbar J$ and $\hbar D_k$ are exchange coupling and onsite potential energies, respectively. We assume that the local onsite potentials can be controlled via local electric or magnetic fields that shift the energy levels of the qubit, e.g., via Stark or Zeeman shifts, respectively, i.e., $D_1$ and $D_2$ serve as controls. To include decoherence, the state of the system is described by the density operator $\rho$, a Hermitian operator with $\Tr(\rho)=1$, which evolves as
\begin{equation}
  \tfrac{d}{dt} \rho = -\tfrac{\imath}{\hbar} [H,\rho] + \gamma \mathfrak{L}[V] \rho,
\end{equation}
where $\mathfrak{L}[V]$ is a typical Lindbladian,
\begin{equation} \label{eq:lindblad}
  \mathfrak{L}[V]\rho=V\rho V^\dagger-\tfrac{1}{2}(V^\dagger V\rho+\rho V^\dagger V).
\end{equation}
To model typical decoherence, which acts as dephasing in the Hamiltonian basis, we choose $V$ to be a Hermitian operator that commutes with the Hamiltonian, $[H,V]=0$. We take
\begin{equation*}
  V = \frac{1}{2 J_{\rm eff}} \begin{bmatrix} -\Delta & 2J \\ 2J & \Delta\end{bmatrix}.
\end{equation*}
where $\Delta := D_2-D_1$ is the detuning and $J_{\rm eff} = \sqrt{\Delta^2+4J^2}$. $\Delta/J$ is dimensionless, $\hbar\Delta$ is the difference between the energy levels of the left and right qubit and $\gamma$ is the decoherence rate.

The dynamics can be reformulated in a way similar to the state space formulation classically used in robustness analysis using the Bloch equation formulation~\cite{PhysRevA.93.063424,PhysRevA.81.062306,JPhysA37,neat_formula}, here rewritten in line with Eqs.~\eqref{eq:r_unperturbed} and~\eqref{eq:r_perturbed},
\begin{subequations}
\begin{align}
  \tfrac{d}{dt} \vec{r}_u &= \underbrace{(\A_H+\A_\mathfrak{L})}_{\A}\vec{r}_u+\bm B\vec{d}, \\
  \tfrac{d}{dt} \vec{r}_p &= \underbrace{(\A_H+\A_\mathfrak{L}+\delta \S)}_{\A+\delta \S}\vec{r}_u+\bm B\vec{d}.
  \end{align}
\end{subequations}
This formulation is obtained by expanding the density operator for the quantum state in terms of the Pauli matrices
\begin{equation} \label{eq:pauli}
 \sigma_x = \begin{bmatrix} 0 &  1\\ 1 &  0 \end{bmatrix},
 \sigma_y = \begin{bmatrix} 0 & -i\\ i &  0 \end{bmatrix},
 \sigma_z = \begin{bmatrix} 1 &  0\\ 0 & -1 \end{bmatrix},
 I= \begin{bmatrix} 1 & 0 \\ 0 & 1 \end{bmatrix}.
\end{equation}
For a two-qubit system it is customary to use the expansion
\begin{equation} \label{eq:rho-expansion}
  \rho=\tfrac{1}{2} (r_1 \sigma_x + r_2 \sigma_y + r_3 \sigma_z + I)
\end{equation}
with regard to the unnormalized Pauli matrices rather than an orthonormal basis for the Hermitian operator on the Hilbert space as this ensures that pure states, characterized by $\Tr(\rho^2)=1$, are mapped to points $[r_1,r_2,r_3]^T$ on the unit sphere in $\mathbb{R}^3$. Due to trace conservation, the coefficient $r_4$ is constant, and, with our choice of basis, $r_4=1$. Therefore, it makes sense to define the reduced state vector $\vec{r}_u=[r_1,r_2,r_3]^T \in \mathbb{R}^3$. It can be shown by direct calculation that the corresponding dynamical generators are explicitly
\begin{equation*}
  \A_H =\begin{bmatrix}
    0 & \Delta & 0 \\
    -\Delta & 0 & -2J \\
    0 & 2J & 0
  \end{bmatrix}, \,
  \A_\mathfrak{L} = \frac{-\gamma}{J_{\rm eff}^2}
  \begin{bmatrix}
    \Delta^2 & 0 & 2\Delta J \\
    0 & J_{\rm eff}^2 & 0 \\
    2\Delta J & 0 & 4J^2
  \end{bmatrix}.
\end{equation*}
The matrix $\A$ has rank $2$ with a pole at $0$ for any $\Delta$, $J$ and $\gamma$. The oscillatory eigenvalues are $\lambda(\A_H+\A_\mathfrak{L})=-\gamma \pm \imath J_{\mathrm{eff}}$. The \emph{structural stability} of the eigenvalues obviates the need to consider possible change of multiplicities, as considered in~\cite{robust_performance_open}.

The excitation in the coupled qubit system undergoes coherent oscillations between the left and right qubit, $\rho_{L}=\ket{L}\bra{L}$ and $\rho_{R}=\ket{R}\bra{R}$, corresponding to $\vec{r}_{L} = [0,0,1]^T$ and $\vec{r}_{R}=[0,0,-1]^T$, for the unperturbed and perturbed cases. Starting with an excitation of the left qubit, $\rho(0)=\rho_L$, we measure the excitation transfer to the right qubit as a function of time by the overlap $\mathcal{F}_{u/p}(t) = \Tr[\rho_R\rho_{u/p}(t)]$, where $\rho_{u/p}$ denotes the unperturbed or perturbed density~\cite{Liang_2019}. For $J$, $\Delta$ and $\gamma$ fixed and no noise (and the affine term can be shown to be zero), the evolution is given by $\vec{r}_u(t) = \exp(t \A) \vec{r}_u(0)$, where $\exp$ denotes the matrix exponential. Setting $\vec{r}_{u/p}(t) = (x(t),y(t),z(t))_{u/p}^T$, this shows that $\mathcal{F}_{u/p}(t) = \tfrac{1}{2}[1+\vec{r}_{R}^T \vec{r}_{u/p}(t)] = \tfrac{1}{2}[1-z_{u/p}(t)]$, and we can explicitly evaluate the matrix exponential to obtain an analytic formula for the transfer fidelity in the unperturbed case,
\begin{equation} \label{eq:q1-fid_t}
  \mathcal{F}_u(t) = \frac{2J^2}{J_{\rm eff}^2} \left[ 1-e^{-t\gamma} \cos(J_{\rm eff} t)\right].
\end{equation}

Eq.~\eqref{eq:q1-fid_t} shows that the fidelity undergoes damped oscillations as illustrated in Fig.~\ref{fig:Quantum1}(a). The first maximum is achieved for $\cos(J_{\rm eff} t) = -1$, i.e., a transfer time $t_f=\pi/J_{\rm eff}$. The corresponding fidelity is
\begin{equation}\label{eq:q1-fid_max}
  \mathcal{F}_{u}(t_f) = \frac{2J^2}{J_{\rm eff}^2} [ 1+ e^{-t_f\gamma}].
\end{equation}
Considering $\gamma\ge 0$, the maximum fidelity is $\mathcal{F}_{u,\max}(t_f) =\frac{4J^2}{J_{\rm eff}^2}$, achieved for $\gamma=0$. 
Recalling the definition for $J_{\rm eff} = \sqrt{\Delta^2+4J^2}$, we further see that the absolute maximum of $1$ is achieved for $\Delta=0$ and $J_{\rm eff}=2J$. There is a trade-off between the maximum transfer fidelity and the transfer time, however, as $t_f=\pi/J_{\rm eff}$ implies that the transfer time decreases with increasing $\Delta$. This offers some advantage in terms of increased transfer fidelity, as illustrated in Fig.~\ref{fig:Quantum1}(b), which shows that both the total fidelity and minimal transfer time decrease with increasing detuning $\Delta$. In most practical cases, however, when full state transfer is desired, the speedup due to nonzero $\Delta$ is minor compared to the fidelity loss. Additionally, while the fidelity starts at a lower value when decoherence is present, in this case it decreases more slowly with increasing detuning due to the decreased transfer time.

\subsection{Classical structured singular value analysis}

The dynamical generator depends on three core parameters $\Delta$, $J$, and $\gamma$, which are often not known precisely. To assess the robustness of observed fidelity oscillations with regard to these parameter uncertainties, we can define structured perturbations. The dependence on $\gamma$ is linear and thus directly amenable to structured singular value analysis. In what follows, we reserve the symbol $\delta$ for general perturbations, and $\delta_{(\ell)}$ with the subscript in parentheses for perturbations on the parameter $\ell = \Delta, \gamma, J$.

The infidelity, either unperturbed or perturbed, is defined as $1-\mathcal{F}_{u/p}(t) = \tfrac{1}{2}[1+z_{u/p}(t)] = \tfrac{1}{2}[0,0,1]\vec{r}_{u/p}+\frac{1}{2}$, which requires an unusual affine term. The infidelity error, on the other hand, defined as
\begin{equation*}
  \vec{e} = (1-\mathcal{F}_p)-(1-\mathcal{F}_u)
          = \mathcal{F}_u-\mathcal{F}_p=\C_u \vec{z}
\end{equation*}
with $\C_u=\tfrac{1}{2}[0,0,1]$ and $\vec{z}=\vec{r}_p-\vec{r}_u$, does not require the affine term. Consistently with Eqs.~\eqref{eq:z_unperturbed} and~\eqref{eq:z_perturbed}, the infidelity error $\vec{e}$ has two formulations: unperturbed and perturbed.

Considering an uncertainty $\gamma(1+\delta_{(\gamma)})$ on the decoherence rate, the structure of the perturbation is $\S=\gamma \A_\mathfrak{L}$. Numerical results for $\Delta=0$, $J=1$ and $\gamma=0.01$, shown in Fig.~\ref{fig:quantum1-gamma-mussv} suggest good coincidence between the unperturbed and perturbed analyses but there are differences. In the unperturbed case the graph suggests $\mu^u_\infty=1.276>1$; hence, $\|\T^u\|_\infty < 1.272$ for $|\delta|<1/1.276=0.78<1$. The restriction $|\delta|< 1$ is consistent with ${\gamma}_p=\gamma(1+\delta)>0$, that is, positive perturbed decoherence rate, and hence robust closed-loop stability. Closer inspection of the upper bound for $\mu^p$ suggests $\|\T^p\|_\infty \leq \mu^p_\infty \approx 0.5946<1$ for $\omega=2$. This bound is tighter than $\mu^u$, but as already noted $\mu^p$ does not secure closed-loop stability so that the bound $\|\T^p\|_\infty \leq 0.5946$ is only valid for $|\delta|<1$.

\subsection{Beyond structured singular value: fixed-point iteration}

The dependence of $\A$ on $\Delta$ and $J$ is nonlinear, which poses a challenge for structured singular value analysis. This illustrates its limitations and why new tools are needed, especially for quantum control problems, where parameters typically enter in a nonlinear fashion in the Bloch $\A$-matrix. In~\cite{CDC2021_mu}, a fixed-point approach is proposed in lieu of the structured singular value. In essence the perturbation is unstructured and written as $\A_p(\delta)-\A$ rather than $\S\delta$ as in Eqs.~\eqref{eq:r_unperturbed}--\eqref{eq:r_perturbed}. In the unperturbed formulation case, $\|\T^{u}_{\vec{e},\vec{w}_{u}}(\delta)\|_\infty\leq \mu$, $\forall \delta<1/\mu$, where the strict inequality is required to enforce strict closed-loop stability. Interchanging the role of $\mu$ and $\delta$ yields $\|\T^{u}_{\vec{e},\vec{w}_{u}}(\delta)\|_\infty < 1/\delta$, leading to the tightening of the inequality by defining
\begin{equation}\label{eq:delta_max}
  \delta_{\mathrm{max}}=\sup\left\{\delta \geq 0: \|\T^{u}_{\vec{e},\vec{w}_{u}}(\delta)\|_\infty < 1/\delta\right\}.
\end{equation}
Whether the \emph{equality} $\|\T^{u}_{\vec{e},\vec{w}_{u}}(\delta_{\max})\|_\infty = 1/\delta_{\max}$ could be achieved depends on whether the limiting factor is closed-loop stability or the transmission norm. In quantum problems as those considered here,
the uncertain parameters do not affect stability, and hence the limiting factor is the transmission norm. In such cases, \emph{equality} prevails, and $\delta_{\max}$ is a fixed point of $\delta \mapsto 1/\|\T^u_{\vec{e},\vec{w}_{u}}(\delta)\|_\infty$. Note that for the solution to Eq.~\eqref{eq:delta_max} to be a fixed point, the reverse inequality $\|\T^u_{\vec{e},\vec{w}_{u}}(\delta)\|_\infty> 1/\delta$ must hold $\forall \delta>\delta_{\max}$. If there are many such fixed points, the minimum one should be selected. Note that this approach identifies the positive bound; the negative bound $\delta_\mathrm{min}<0$ can be identified in a similar way via the fixed point of $\delta \mapsto -1/\|\T^u_{\vec{e},\vec{w}_{u}}(\delta)\|_\infty$.

In the unperturbed and perturbed cases, the fixed point can be computed by graphing $\|\T^u_{\vec{e},\vec{w}_{u}}(\delta) \|_\infty$ and $1/\delta$ vs.\ $\delta$ and locating the intersection point of the two plots. More formally, under some circumstances, the contraction mapping theorem can be invoked leading to convergence of the recursion $\delta_{k+1}=1/\|\T^u_{\vec{e},\vec{w}_{u}}(\delta_k)\|_\infty$ to the fixed point. Contraction mapping requires existence of a contraction ratio $\zeta \in (0,1)$ such that
\begin{equation*}
  |\delta_{k+2}-\delta_{k+1}|< \zeta |\delta_{k+1}-\delta_k|.
\end{equation*}
That is, simplifying the notation to avoid the clutter,
\begin{equation*}
  \left|\left\|T\left(\|T(\delta)\|^{-1}\right)\right\|^{-1}
  -\|T(\delta)\|^{-1}
  \right| < \zeta \left|\|T(\delta)\|^{-1}-\delta\right|.
\end{equation*}
From there, it is easily verified that, for $m>n$,
\begin{equation*}
  |\delta_m-\delta_n| < \frac{\zeta^m-\zeta^{n-1}}{\zeta-1}|\delta_1-\delta_0|.
\end{equation*}
In other words, $\{\delta_k\}_{k=0}^\infty$ is a Cauchy sequence and, hence, converges.

\subsubsection{Decoherence rate uncertainty}

\begin{figure}
\centering
\includegraphics[width=\columnwidth]{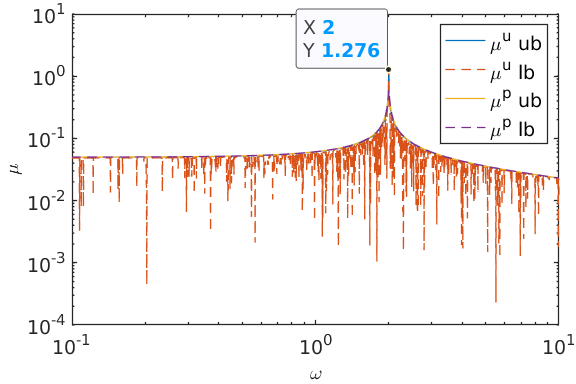}
\caption{Upper (ub) and lower (lb) bounds on unperturbed and perturbed $\mu$'s for two qubits under decoherence rate uncertainty (system parameters $\Delta =0$, $J=1$, $\gamma=0.01$).}\label{fig:quantum1-gamma-mussv}
\end{figure}

\begin{figure*}
\subfloat[Unperturbed Transfer function]{\includegraphics[width=0.32\textwidth]{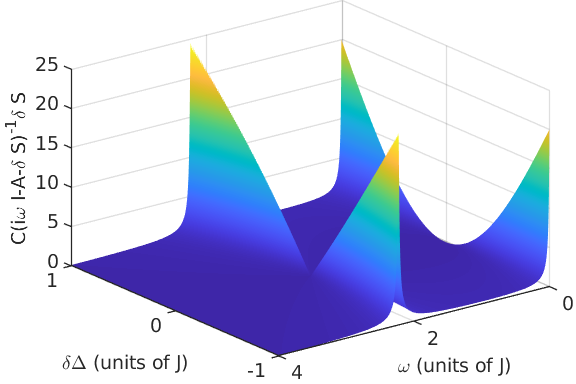}}
\hfill
\subfloat[Perturbed Transfer function]{\includegraphics[width=0.32\textwidth]{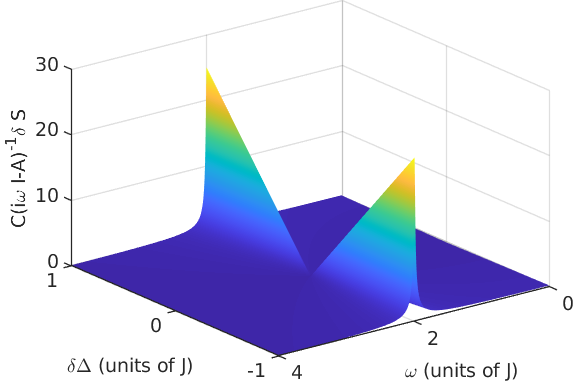}}
\hfill
\subfloat[$\mu$ vs $\delta$]{\includegraphics[width=0.32\textwidth]{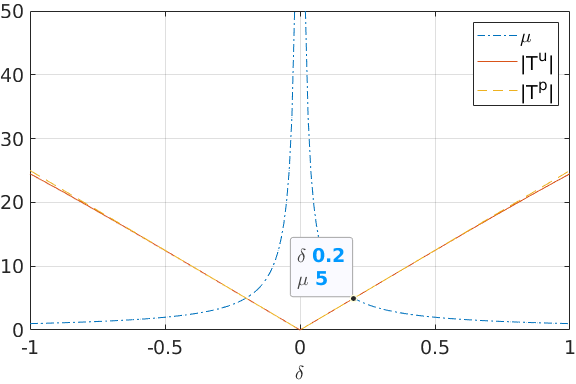}}
\caption{Transfer function relative to unperturbed and perturbed state and $\mu$ for uncertainty in $\Delta=\Delta_0+\delta_{(\Delta)}$ with $\Delta_0=0$, $J=1$, $\gamma=0.01$. The intersection point of $\delta_{(\Delta)} = \mu(\delta_{(\Delta)})$ determines $\delta_{(\Delta),\max}$ and $\mu_{(\Delta),\infty}$.}\label{fig:Quantum1-Delta}
\end{figure*}

\begin{figure*}
\subfloat[Unperturbed transfer function]{\includegraphics[width=0.32\textwidth]{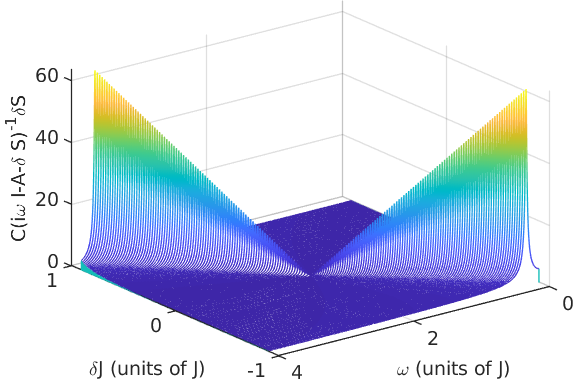}}
\hfill
\subfloat[Perturbed transfer function]{\includegraphics[width=0.32\textwidth]{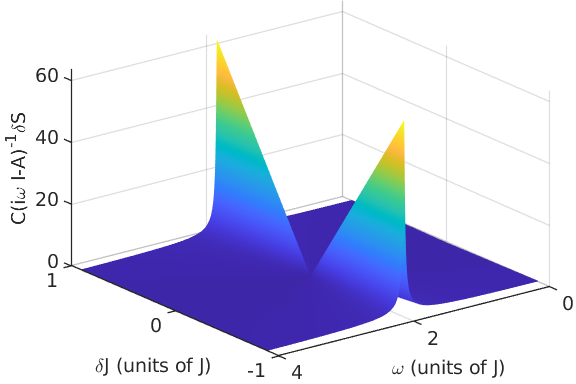}}
\hfill
\subfloat[$\mu$ vs $\delta$]{\includegraphics[width=0.32\textwidth]{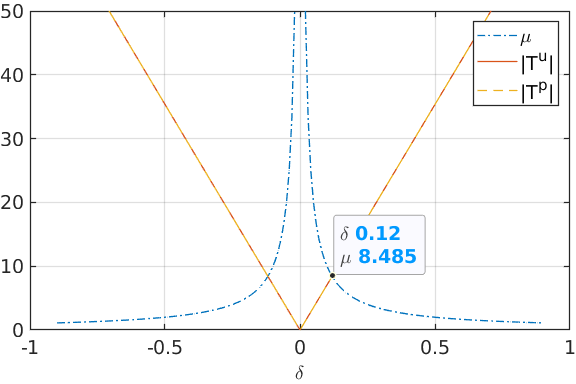}}
\caption{Transfer function relative to unperturbed and perturbed state and $\mu$ for uncertainty in $J=J_0+\delta_{(J)}$ with $J_0=1$, $\Delta=0$, $\gamma=0.01$. The intersection point of $\delta_{(J)} = \mu(\delta_{(J)})$ determines $\delta_{(J),\max}$ and $\mu_{(J),\infty}$.}\label{fig:Quantum1-J}
\end{figure*}

\begin{figure*}
\subfloat[Unperturbed transfer function]{\includegraphics[width=0.32\textwidth]{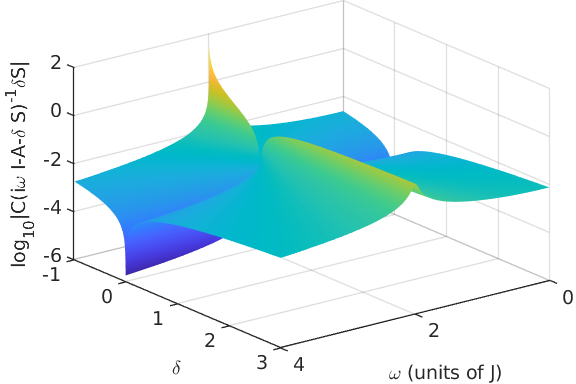}}
\hfill
\subfloat[Perturbed transfer function]{\includegraphics[width=0.32\textwidth]{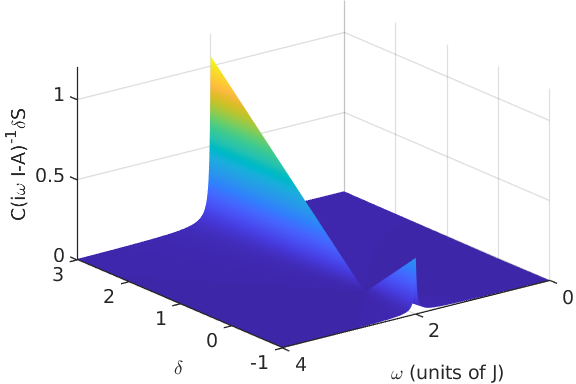}}
\hfill
\subfloat[$\mu$ vs $\delta$]{\includegraphics[width=0.32\textwidth]{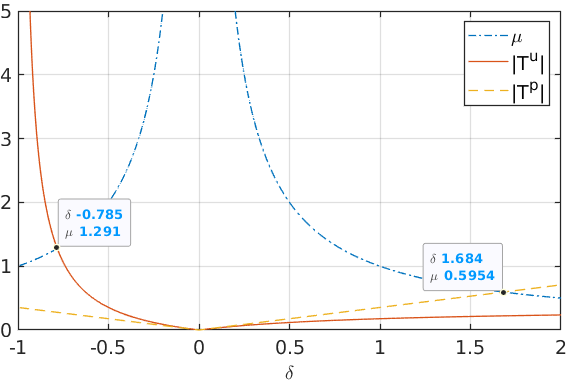}}
\caption{Transfer function relative to unperturbed and perturbed state and $\mu$ for uncertainty in $\gamma=\gamma_0+\delta_{(\gamma)}$ with $\gamma_0=0.01$, $J=1$, $\Delta=0$. The intersection point of $\delta_{(\gamma)} = \mu(\delta_{(\gamma)})$ determines $\delta_{(\gamma), \max}$ and $\mu_{(\gamma),\infty}$.}\label{fig:Quantum1-gamma}
\end{figure*}

Since the case of an uncertain decoherence rate $\gamma$ is the only one where the uncertain parameter enters linearly into the Bloch equation and can be safely computed with classical $\mu$, we will use that case as yardstick to gauge how well the fixed point approach works.
Taking $\Delta=0$, $J=1$, and $\gamma=0.01$ as before yields $\mu^u=0.277$ at $\delta_{\max}= 3.6$ and $\mu^u = 1.282$ at $\delta_{\mathrm{min}}=-0.78$ as shown in Fig.~\ref{fig:Quantum1-gamma}(c), consistent with the upper bound for $\mu_\infty^u$ obtained from the conventional $\mu$-analysis and Fig.~\ref{fig:quantum1-gamma-mussv}.

\begin{table}
\centering
\begin{tabular}{|l|cc|cc|cc|}\hline
  & \multicolumn{2}{c|}{$\Delta = \Delta_0 + \delta_{(\Delta)}$}
    & \multicolumn{2}{c|}{$J=J_0+\delta_{(J)}$}
    & \multicolumn{2}{c|}{$\gamma = \gamma_0(1+\delta_{(\gamma)})$} \\\hline
  & $\delta_{\min}$ & $\delta_{\max}$ & $\delta_{\min}$ & $\delta_{\max}$ & $\delta_{\min}$ & $\delta_{\max}$\\\hline
  $\T^u$ & -0.200 & 0.200 & -0.1194 & 0.1194 & -0.7832 & 3.6114 \\\hline
  $\T^p$ & -0.200 & 0.200 & -0.1189 & 0.1189 & -1.6818 & 1.6818 \\\hline
  $\T^u$ & -0.3452 & 0.3452 & -0.3759 & 0.3759 & -0.7832 & 3.5925\\\hline
  $\T^p$ & -0.6315 & 0.6315 & -0.3760 & 0.3760 & -1.6815 & 1.6815 \\\hline
\end{tabular}
\caption{$\delta_{\min}$ and $\delta_{\max}$ obtained by fixed point iteration for perturbations of $\Delta$, $J$ and $\gamma$ for our quantum example. The nominal (unperturbed) values are $\Delta_0=0$, $J_0=1$ and $\gamma_0=0.01$ for rows one and two; $\Delta_0=0$, $J_0=1$ and $\gamma_0=0.1$ for rows three and four.} \label{table:delta}
\end{table}

\subsubsection{Detuning and $J$-coupling uncertainty}

The advantage of the fixed-point iteration approach is that it can applied to assess the effect of uncertainty in the detuning $\Delta$ or $J$-coupling although these parameters appear non-linearly in the dynamics. Figs.~\ref{fig:Quantum1-Delta} and~\ref{fig:Quantum1-J} show the perturbed and unperturbed transfer functions and $\mu$ for uncertainty in $\Delta$ and $J$, respectively, expanded around $(\Delta,J,\gamma)=(0,1,0.01)$. Although there are differences between the transfer functions relative to the unperturbed and perturbed state, both transfer functions have almost the same norm for a wide range of $\delta$ and give the same $\delta_{\min}$ and $\delta_{\max}$, as illustrated in Table~\ref{table:delta}. However, the critical frequencies of $\T^u$ and $\T^p$ are different, especially for perturbation of $J$, as expected. Comparing the $\mu$ plots suggests that $\mu_\infty=\delta_{\max}^{-1}$ is larger for $J$-coupling uncertainty, which suggests that the system is more sensitive to perturbation of $J$ than $\Delta$. The $\mu_\infty$ values for uncertainty in $\Delta$ and $J$ are also much larger than $\mu_\infty$ for $\gamma$ uncertainty, suggesting that the system is far less affected by uncertainty in the decoherence rates. In practice, this analysis is useful for understanding where the fundamental limitations of real systems lie and where extra care must be placed when designing a quantum system.

\subsection{Quantum vs. mechanical systems}

As already said, a limitation to passivity-control of lightly damped oscillatory mechanical systems is the co-location error between point of application of actuators and point of rate recording. Quantum systems also involve co-location at their core. In the single-particle Hamiltonian $H=\tfrac{1}{2}\omega_q \sigma_z$, the Pauli operator $\sigma_z$ indicates that the qubit is subject to a magnetic field along the $z$-axis and measurements with possible outcomes spin up or down are relative to \emph{exactly} the same axis. Spin chains involve co-location errors, referred to as \emph{bias spillage}~\cite{soneil_mu} along the chain axis, conceptually similar to co-location error in space structures. A highly focused magnetic field meant to address a single spin always entails an error between the bias field and the spin it is supposed to address; a caveat that is well known in passivity control of distributed parameter space structures~\cite{positive_real}.

Besides this simple quantum-classical oscillatory system similarity, 
the two classes of systems differ in a much more profound way. 
Left-half plane pole shifting in classical mechanical systems means that the system dissipates energy, 
a concept fundamental in Lyapunov stability, 
positive-real design (Sec.~\ref{s:positive_real}), even IQCs. 
However, in a quantum system like the present one where dephasing shifts the poles to the left-half plane, 
the system does not dissipate energy, but loses information. This can be quantified by an increase of von Neumann entropy, $-\mbox{Trace}(\rho\log\rho)$, which departs from the quadratic methods of IQCs. 

\section{Conclusion and further prospects}

The key point of this paper is that robust performance under structured uncertainties in both classical and quantum systems can be assessed independently of additive disturbances and/or noises, 
which are sometimes physically legitimate but too often solely motivated by a textbook solution 
to robust performance assessment. 
If adding disturbances is legitimate, then the proposed approach has the advantage of homing in on the effect of the uncertain parameters independently of the disturbances.

Suppression of the disturbance is accomplished by removing the unperturbed dynamics from the perturbed dynamics leading to two different but equivalent error systems: one driven by the unperturbed dynamics, the other by the perturbed dynamics. From a design point of view, the unperturbed forcing term calls into question the relevance of the worst case, or $H^\infty$, approach since the driving term is perfectly known, possibly resuscitating the old geometric disturbance decoupling problem. On the other hand, for the perturbed driven model, the $H^\infty$ approach might still, in spirit, be appropriate since the disturbance is imprecisely known, yet not totally unknown. 
These design questions are left for further research.

Also left for further research is the development of a genuine Brouwer fixed point substitute for the $\mu$-function mandated by nonlinear uncertainties. Last but not least, an information loss approach to quantum dephasing systems that parallels the energy loss approach in classical system is another future challenge.

\bibliographystyle{plain}
\bibliography{biblio/edmond,biblio/physics,biblio/power_grid,biblio/oxford}

\begin{thebibliography}{10}

\bibitem{Belevitch}
V.~Belevitch.
\newblock {\em Classical Network Theory}.
\newblock Holden-Day, San Francisco, 1968.

\bibitem{Carrie_Kosut}
Yulong Dong, Xiang Meng, Lin Lin, Robert Kosut, and K.~Birgitta Whaley.
\newblock Robust control optimization for quantum approximate optimization
  algorithms.
\newblock {\em 21st IFAC World Congress}, 53(2):242--249, 2020.

\bibitem{Ezawa1}
M.~Ezawa.
\newblock Electric-circuit simulation of the {S}chr\"odinger equation and
  non-{H}ermitian quantum walks.
\newblock {\em Phys. Rev. B}, 100:165419, 2019.

\bibitem{neat_formula}
F.~F. Floether, P.~de~Fouquieres, and S.~G. Schirmer.
\newblock Robust quantum gates for open systems via optimal control:
  {M}arkovian versus non-{M}arkovian dynamics.
\newblock {\em New Journal of Physics}, 14:1--26, 2012.
\newblock 073023.

\bibitem{CompetitiveControl}
Gautam Goel and Babak Hassibi.
\newblock Competitive control.
\newblock arXiv:2107.13657 [math.OC], 2021.

\bibitem{Safonov89}
A.~Holohan and M.~G. Safonov.
\newblock On computing the {MIMO} real structured stability margin.
\newblock In {\em Proc. IEEE Conf. on Decision and Control}, Tampa, FL,
  December 13-15 1989.

\bibitem{chains_QINP}
E.~A. Jonckheere, S.~G. Schirmer, and F.~C. Langbein.
\newblock Quantum networks: {T}he anti-core of spin chains.
\newblock {\em Quantum Information Processing}, 13:1607--1637, 2014.

\bibitem{rings_QINP}
E.~A. Jonckheere, S.~G. Schirmer, and F.~C. Langbein.
\newblock Information transfer fidelity in spin networks and ring-based quantum
  routers.
\newblock {\em Quantum Information Processing}, 14(10):4751--4785, 2015.

\bibitem{statistical_control}
E.~A. Jonckheere, S.~G. Schirmer, and F.~C. Langbein.
\newblock Jonckheere-{T}erpstra test for nonclassical error versus
  log-sensitivity relationship of quantum spin network controllers.
\newblock {\em International Journal of Robust and Nonlinear Control},
  28(6):2383--2403, 2018.

\bibitem{PRA}
Irtaza Khalid, Carrie~A. Weidner, Edmond~A. Jonckheere, Sophie~G. Shermer, and
  Frank~C. Langbein.
\newblock Statistically characterizing robustness and fidelity of quantum
  controls and quantum control algorithms.
\newblock {\em Phys. Rev. A}, 107:032606, Mar 2023.

\bibitem{Liang_2019}
Y.-C. Liang, Y.-H. Yeh, P.~E. M.~F. Mendon{\c{c}}a, R.~Y. Teh, M.~D. Reid, and
  P.~D. Drummond.
\newblock Quantum fidelity measures for mixed states.
\newblock {\em Reports on Progress in Physics}, 82(7):076001, 2019.

\bibitem{Safonov95}
J.~H. Ly, K.~C. Goh, and M.~G. Safonov.
\newblock Multiplier $k_m$/$\mu$-synthesis — {LMI} approach.
\newblock In {\em In Proc. American Control Conf.}, page 431–346, Seattle,
  WA, June 21–23 1995.

\bibitem{soneil_mu}
S.~O'Neil, E.~Jonckheere, S.~Schirmer, and F.~Langbein.
\newblock Sensitivity and robustness of quantum rings to parameter uncertainty.
\newblock In {\em IEEE CDC}, pages 6137--6143, 2017.

\bibitem{symmetric_passive}
{PH}. Opdenacker and E.~A. Jonckheere.
\newblock {LQG} balancing and reduced {LQG} compensation of symmetric passive
  systems.
\newblock {\em International Journal of Control}, 41(1):73--109, 1985.

\bibitem{TRW}
{PH}.~C. Opdenacker, E.A. Jonckheere, M.G. Safonov, J.C. Juang, and M.~Lukich.
\newblock Reduced order compensator design of a flexible structure.
\newblock {\em Journal of Guidance, Control and Dynamics}, 13(1):46--56, 1990.

\bibitem{PackardDoyle}
Andrew Packard and John~C. Doyle.
\newblock The complex structured singular value.
\newblock {\em Autom.}, 29:71--109, 1993.

\bibitem{PhysRevA.93.063424}
P.~Rooney, A.~M. Bloch, and C.~Rangan.
\newblock Flag-based control of quantum purity for $n=2$ systems.
\newblock {\em Phys. Rev. A}, 93:063424, 2016.

\bibitem{Safonov84}
M.~G. Safonov.
\newblock Exact calculation of the multivariable structured-singular-value
  stability margin.
\newblock In {\em Proc. IEEE Conf. on Decision and Control}, Las Vegas, NV,
  December 12–14 1984.

\bibitem{positive_real}
M.~G. Safonov, E.~A. Jonckheere, M.~Verma, and D.~J.~N. Limebeer.
\newblock Synthesis of positive real multivariable feedback systems.
\newblock {\em International Journal of Control}, 45(3):817--842, 1987.

\bibitem{Safonov_Laub_Hartmann}
M.~G. Safonov, A.~J. Laub, and G.~L. Hartmann.
\newblock Feedback properties of multivariable systems: The role and use of the
  return difference matrix.
\newblock {\em IEEE Transactions on Automatic Control}, 26(1):47--65, 1981.

\bibitem{Edmond_IEEE_AC}
S.~G. Schirmer, E.~A. Jonckheere, and F.~C. Langbein.
\newblock Design of feedback control laws for spintronics networks.
\newblock {\em IEEE Transactions on Automatic Control}, 63(8):2523--2536, 2018.

\bibitem{CDC_decoherence}
S.~G. Schirmer, E.~A. Jonckheere, S.~O'Neil, and F.~C. Langbein.
\newblock Robustness of energy landscape control for spin networks under
  decoherence.
\newblock In {\em IEEE CDC}, pages 6608--6613, 2018.

\bibitem{CDC2021_mu}
S.~G. Schirmer, F.~C. Langbein, C.~A. Weidner, and E.~A. Jonckheere.
\newblock Robustness of quantum systems subject to decoherence: Structured
  singular value analysis?
\newblock In {\em IEEE CDC}, pages 4158--4163, 2021.

\bibitem{robust_performance_open}
S.~G. Schirmer, F.~C. Langbein, C.~A. Weidner, and E.~A. Jonckheere.
\newblock Robust control performance for open quantum systems.
\newblock {\em IEEE Transactions on Automatic Control}, 67(11):6012--6024,
  November 2022.

\bibitem{PhysRevA.81.062306}
S.~G. Schirmer and Xiaoting Wang.
\newblock Stabilizing open quantum systems by {M}arkovian reservoir
  engineering.
\newblock {\em Phys. Rev. A}, 81:062306, 2010.

\bibitem{JPhysA37}
S.~G. Schirmer, T.~Zhang, and J.~V. Leahy.
\newblock Orbits of quantum states and geometry of {B}loch vectors for
  {N}-level systems.
\newblock {\em Journal of Physics A: Mathematical and General}, 37(4):1389,
  2004.

\bibitem{Tempo2013UncertainLS}
Roberto Tempo, Giuseppe~Carlo Calafiore, and Fabrizio Dabbene.
\newblock Uncertain linear systems.
\newblock 2013.

\bibitem{IQC}
Joost Veenman, Carsten~W. Scherer, and Hakan Köroğlu.
\newblock Robust stability and performance analysis based on integral quadratic
  constraints.
\newblock {\em European Journal of Control}, 31:1--32, 2016.

\bibitem{topological_insulators}
Y.~Wang, H.~M. Price, B.~Zhang, and Y.~D. Chong.
\newblock {Circuit implementation of a four-dimensional topological insulator}.
\newblock {\em Nature Communications}, 11:2356, 2020.

\bibitem{Benchmark}
B.~Wie and D.~Bernstein.
\newblock Benchmark problem for robust control design.
\newblock {\em American Institute of Aero. and Astro.}, 15:1057--1059, 09 1992.

\bibitem{Young2001StructuredSV}
Peter~Michael Young.
\newblock Structured singular value approach for systems with parametric
  uncertainty.
\newblock {\em International Journal of Robust and Nonlinear Control}, 11,
  2001.

\bibitem{Zhou}
K.~Zhou and J.~C. Doyle.
\newblock {\em Essentials of Robust Control}.
\newblock Prentice Hall, Upper Saddle River, NJ, 1998.

\end{thebibliography}

\end{document}